\newif\ifarxiv
\arxivtrue
\RequirePackage{fix-cm}
\ifarxiv
  \documentclass[11pt]{article}
\else
  \documentclass[smallextended,a4paper,runningheads]{svjour3}
\fi

\ifarxiv
  \usepackage{amsthm}
  \newtheorem{theorem}{Theorem}
  \newtheorem{lemma}{Lemma}
  \newtheorem{corollary}{Corollary}
  \newtheorem{observation}{Observation}
  \theoremstyle{definition}
  \newtheorem{definition}{Definition}
  \newtheorem{claim}{Claim}
  \newcommand{\claimqed}{\renewcommand{\qedsymbol}{$\triangle$}}
  \usepackage{authblk}
  \newcommand{\keywords}[1]{%
    \par\medskip
    \noindent\textbf{Keywords: }{\def\and{, }#1}
  }
  \usepackage[a4paper,left=3cm,right=3cm]{geometry}
  \usepackage{microtype}
\else
  \smartqed
  \usepackage{etoolbox}
  \AtEndEnvironment{proof}{\qed} % Automatically appends the QED symbol to all proofs
  \usepackage{lineno}
  
  \spnewtheorem{claim}{Claim}{\bfseries}{\rmfamily}
  \spnewtheorem{observation}{Observation}{\bfseries}{\itshape}
  \newcommand{\claimqed}{\renewcommand{\squareforqed}{\hbox{$\triangle$}}}
\fi

\usepackage{amssymb}
\usepackage{amsmath}
\usepackage{enumitem}
\usepackage{booktabs}
\usepackage{array}
\usepackage[full]{complexity}
\usepackage[linesnumbered,lined,ruled,noend,vlined]{algorithm2e}
\usepackage[round]{natbib}
\setcitestyle{aysep={}}
\usepackage{hyperref}

\ifarxiv
\else
  \makeatletter
  \let\cl@chapter\relax
  \makeatother
\fi

\usepackage{cleveref}
\usepackage{orcidlink}
\usepackage{xspace}
\usepackage{xstring}
\usepackage{tcolorbox}

\newcommand{\defproblem}[4]{
  \begin{tcolorbox}%
    \vspace*{-1ex}\hspace*{-2.5ex}
    \begin{minipage}{0.98\textwidth}
      \begin{tabular}{@{}>{\normalsize}l@{~~}>{\normalsize}p{0.9\textwidth}@{}}
        {\bfseries Problem:} & #1\\[.1ex]
        {\bfseries Input:} & #2\\[.1ex]
        {\bfseries #4:} & #3
      \end{tabular}
    \end{minipage}\vspace*{-1ex}
  \end{tcolorbox}
}

\newcommand{\defoptproblem}[3]{\defproblem{#1}{#2}{#3}{Output}}

\SetKwProg{Procedure}{Procedure}{}{}

\newcommand{\set}[1]{\{#1\}}
\newcommand{\gged}[1][]{%
  \textup{\textsc{GGED}\if\relax\detokenize{#1}\relax\else($#1$)\fi}\xspace
}
\newcommand{\ggedlong}{\textsc{Geometric Graph Edit Distance}\xspace}

\providecommand{\E}{}
\renewcommand{\E}{\mathcal{E}}
\newcommand{\I}{\mathcal{I}}
\newcommand{\J}{\mathcal{J}}
\renewcommand{\S}{\mathcal{S}}

\newcommand{\complete}{\ensuremath{\Pi_{\mathrm{comp}}}\xspace}
\newcommand{\edgeless}{\ensuremath{\Pi_{\mathrm{edgeless}}}\xspace}
\newcommand{\acyclic}{\ensuremath{\Pi_{\mathrm{acyc}}}\xspace}
\newcommand{\kclique}[1][k]{\ensuremath{\Pi_{#1\mathrm{\text{-}clique}}}\xspace}
\newcommand{\nokclique}[1][k]{\ensuremath{\overline{\Pi_{#1\mathrm{\text{-}clique}}}}\xspace}
\newcommand{\kconnected}[1][k]{\ensuremath{\IfEqCase{#1}{%
        {1}{\Pi_{\mathrm{conn}}}}[\Pi_{#1\mathrm{\text{-}conn}}]}\xspace}

\hypersetup{
  pdftitle={Algorithms for Optimally Shifting Intervals under Intersection Graph Models},
  pdfauthor={Nicol{\'a}s Honorato-Droguett, Kazuhiro Kurita, Tesshu Hanaka and Hirotaka Ono}
}

\crefname{algocf}{algorithm}{algorithms}
\Crefname{algocf}{Algorithm}{Algorithms}
\crefname{AlgoLine}{line}{lines}
\Crefname{AlgoLine}{Line}{Lines}
\crefname{theorem}{theorem}{theorems}
\Crefname{theorem}{Theorem}{Theorems}
\crefname{lemma}{lemma}{lemmas}
\Crefname{lemma}{Lemma}{Lemmas}
\crefname{corollary}{corollary}{corollaries}
\Crefname{corollary}{Corollary}{Corollaries}
\crefname{definition}{definition}{definitions}
\Crefname{definition}{Definition}{Definitions}
\crefname{claim}{claim}{claims}
\crefname{observation}{observation}{observations}
\Crefname{observation}{Observation}{Observations}
\Crefname{claim}{Claim}{Claims}

\ifarxiv\else
\linenumbers
\fi
\begin{document}

\title{Algorithms for Optimally Shifting Intervals under Intersection Graph Models\thanks{A preliminary version \citep*{HonoratoDroguett2024} of this paper appeared in \emph{Proceedings of Frontiers of Algorithmics (IJTCS-FAW 2024)}.}}

\ifarxiv

  \author[1]{Nicol{\'a}s Honorato-Droguett~\orcidlink{0009-0005-1969-3649}}
  \author[2]{Kazuhiro Kurita~\orcidlink{0000-0002-7638-3322}}
  \author[3]{Tesshu Hanaka~\orcidlink{0000-0001-6943-856X}}
  \author[1]{Hirotaka Ono~\orcidlink{0000-0003-0845-3947}}

  \affil[1]{Department of Mathematical Informatics, Nagoya University, Japan\\
    \texttt{honorato.droguett.nicolas.n7@s.mail.nagoya-u.ac.jp, ono@i.nagoya-u.ac.jp}}

  \affil[2]{Information Technology Advanced Track, Okayama University, Japan\\
    \texttt{k-kurita@okayama-u.ac.jp}}

  \affil[3]{Department of Informatics, Kyushu University, Japan\\
    \texttt{hanaka@inf.kyushu-u.ac.jp}}

  \date{}

\else

  \author{
    Nicol{\'a}s Honorato-Droguett~\orcidlink{0009-0005-1969-3649}
    \and
    Kazuhiro Kurita~\orcidlink{0000-0002-7638-3322}
    \and
    Tesshu Hanaka~\orcidlink{0000-0001-6943-856X}
    \and
    Hirotaka Ono~\orcidlink{0000-0003-0845-3947}
  }

  \authorrunning{
    Nicol{\'a}s Honorato-Droguett, Kazuhiro Kurita,
    Tesshu Hanaka, and Hirotaka Ono
  }

  \institute{
    Nicol{\'a}s Honorato-Droguett (corresponding author) \at
    Department of Mathematical Informatics,
    Nagoya University,
    Japan\\
    \email{honorato.droguett.nicolas.n7@s.mail.nagoya-u.ac.jp}
    \and
    Hirotaka Ono \at
    Department of Mathematical Informatics,
    Nagoya University,
    Japan\\
    \email{ono@i.nagoya-u.ac.jp}
    \and
    Kazuhiro Kurita \at
    Information Technology Advanced Track,
    Okayama University,
    Japan\\
    \email{k-kurita@okayama-u.ac.jp}
    \and
    Tesshu Hanaka \at
    Department of Informatics,
    Kyushu University,
    Japan\\
    \email{hanaka@inf.kyushu-u.ac.jp}
  }

  \date{Received: date / Accepted: date}

\fi

\maketitle
\begin{abstract}
  We introduce {\ggedlong}, a model for modifying geometric intersection graphs by moving objects.
  The objective is to obtain a graph in a prescribed class while minimising the total (weighted) moving distance of the objects.
  For weighted intervals of arbitrary length, we give a prune-and-search algorithm that renders their intersection graph complete in linear time.
  We adapt this approach to maximise the common intersection length of weighted intervals within a movement threshold in linear time.
  We also extend the algorithm for obtaining complete graphs to weighted axis-parallel rectangles of arbitrary size under the $L_1$ distance.
  For unweighted unit intervals, we give an $O(n\log n)$-time algorithm for obtaining a graph containing a $k$-clique.
  Lastly, we give $O(n)$-size linear programs for obtaining graphs in several classes, including edgeless, acyclic, $k$-clique-free and $k$-connected graphs.
\end{abstract}


\section{Introduction}
In \emph{graph modification problems} (GMPs), we add or delete vertices or edges to obtain a graph in a given class while minimising the cost of the edit operations.
These problems have applications in computer vision \citep*{Chung1994}, network interdiction \citep*{Hoang2023}, and molecular biology \citep*{Hellmuth2020}.
We distinguish vertex and edge modification problems according to the allowed edit operations.

Although the cost of an edit operation depends on the application, theoretical studies often assume unit costs.
Even under this assumption, many graph modification problems remain \NP-hard \citep*{Lewis1980,Burzyn2006,Fomin2015,Sritharan2016}.
These hardness results motivate alternative formulations that consider constraints and cost models specific to a particular domain.

\emph{Geometric intersection graphs} (hereafter intersection graphs) are a representative example that models relationships induced by geometric overlap, with applications in network modelling and bioinformatics \citep*{McKee1999}.
Formally, given a tuple $\S$ of non-empty subsets of $\mathbb{R}^d$, the intersection graph $G_\S$ has one vertex for each entry of $\S$, and an edge between two distinct vertices if and only if their corresponding objects intersect.
This model includes interval graphs and disk graphs, where the objects are intervals and disks, respectively.

The choice of edit operations and costs is important because it affects both the model and its computational complexity.
As in string similarity analysis, where edit operations can reflect DNA mutations and repeats \citep*{Li2009}, graph modification models should reflect the constraints of their applications.
For spatial applications, intersection graphs provide a suitable framework for studying GMPs, as explored in recent work \citep*{Panolan2024,Berg2019,fomin2023,Fomin2025,Fomin2024-2}.
This motivates two natural questions about GMPs on intersection graphs:
\begin{enumerate}[noitemsep,nosep]
  \item Under the intersection graph model, are edit operations based solely on the graph representation adequate?
  \item How can the geometry of the objects help us design efficient algorithms?
\end{enumerate}
To address these questions, we introduce {\ggedlong}, a model that modifies intersection graphs by moving their objects.
Moving objects can have different objectives, such as removing unwanted overlaps or establishing connections.
For example, in scheduling problems jobs represented as intervals may need to be shifted to remove overlaps, yielding an {\em edgeless} intersection graph.
By contrast, applications requiring connectivity, such as wireless networks, may involve relocating routers represented by disks in the plane so that their intersection graph becomes {\em connected}.

Let $\S=(S_1,\ldots,S_n)$ be an $n$-tuple of geometric objects and $\Pi$ be a graph class.
For a geometric object $S$ and a vector $d$, we write $S+d=\{s+d : s\in S\}$.
For $D=(d_1,\ldots,d_n)$, we write $\S+D=(S_1+d_1,\ldots,S_n+d_n)$ and call $D$ a \emph{distance vector}.
We wish to find, among all distance vectors $D$ such that $G_{\S+D}\in\Pi$, one that minimises total weighted moving distance.
More precisely, we define the problem as follows:
\defoptproblem{\ggedlong w.r.t.\ graph class~$\Pi$}%
{An $n$-tuple~$\S=(S_1,\dots,S_n)$ of geometric objects in
  $\mathbb{R}^d$, weights $w \in \mathbb{R}^n_{> 0}$.}%
{A distance vector $D=(d_1,\dots,d_n) \in \mathbb{R}^{d \cdot n}$ such that
  $G_{\S+D} \in \Pi$ and $\sum_{i=1}^n w_i \cdot \|d_i\|$
  is minimum.}

Here $\|\cdot\|$ is a given norm, which in this paper is the $L_1$ or $L_2$ norm.
We use \gged[\Pi] as shorthand for the above problem.
We primarily focus on tuples of (unit) intervals, whose intersection graphs are called \emph{(unit) interval graphs} and are formally defined in \Cref{sec:preliminaries}.

\subsection{Related Work}

\citet*{Lewis1980} showed that vertex-deletion problems are \NP-complete for obtaining a graph in any nontrivial hereditary graph class recognisable in polynomial time.
Many edge modification problems are also \NP-complete \citep*{Shamir2004,Burzyn2006}.
These hardness results have motivated the study of parameterised algorithms \citep*{Cai1996,Drange2015}, and the survey by \citet*{Crespelle2023} provides an overview of edge modification problems.

\paragraph{Geometric Graph Modification} Movement minimisation has been studied in both graph and geometric settings.
\citet*{Demaine2009} gave approximation algorithms and inapproximability results for moving agents to achieve properties such as connectivity and independence while minimising the total or maximum movement.
Subsequently, \citet*{Demaine2014} studied the parameterised complexity of this framework.
In the Euclidean plane, \citet*{Anari2016} studied movement to obtain a connected unit disk graph or realise a prescribed topology, considering both total and maximum movement.

For points on a line, \citet*{Ghadiri2016} gave an $O(n\log n)$-time algorithm that moves points such that their distance is at least a given threshold, while minimising the total movement.
With threshold~$1$, this corresponds to \gged[\edgeless] on unweighted open unit intervals.
A related interval problem is {\em barrier coverage}, for which \citet*{Andrews2017} gave an $O(n\log n)$-time algorithm for moving unit intervals to cover a given line segment while minimising total movement.

Other geometric graph modification problems have been studied in terms of dispersal, repacking and scaling \citep*{fomin2023,Fomin2024-2,Fomin2025}.
In particular, \citet*{fomin2023} studied the {\em disk dispersal} problem, where the goal is to make open unit disks pairwise disjoint by moving at most $k$ disks by distance at most $d$ each.
They showed that the problem is {\FPT} when parameterised by $k+d$.
The problem is also known to be {\NP-hard} for $d=2$ and $k=n$ \citep*{Fiala2005}.
The same authors studied disk scaling as an edit operation and obtained parameterised algorithms for achieving edgelessness, acyclicity and connectivity \citep*{Fomin2025}.

\paragraph{Geometric Median} The \emph{geometric median} is closely related to our model for obtaining complete graphs.
For points $x_1,\ldots,x_n\in\mathbb{R}^d$ with positive multiplicative weights $w_1,\ldots,w_n$, the geometric median is a point $x$ that minimises $\sum_{i=1}^n w_i\lVert x-x_i\rVert_2$, where $\lVert\cdot\rVert_2$ is the $L_2$ norm.
Weiszfeld \citep*{Weiszfeld2009} proposed an iterative algorithm for this problem.
Although an exact solution cannot in general be expressed by radicals over the rationals \citep*{Bajaj1988}, \citet*{Cohen2016} gave a randomised $(1+\epsilon)$-approximation algorithm that runs in $O(nd\log^3(1/\epsilon))$ time for the weighted case.
We discuss further the relation with \gged[\complete] at the end of \Cref{sec:algorithm}.

\paragraph{Subsequent Work}
Following the preliminary version of this paper \citep*{HonoratoDroguett2024}, we gave $O(n\log n)$-time combinatorial algorithms for obtaining edgeless, acyclic and $k$-clique-free graphs from unweighted unit intervals \citep*{HonoratoDroguett2025}.
These algorithms use the characterisations of our LP constraints and improve the runtime obtained here for these three classes.
Further work \citep*{HonoratoDroguett2026} extended these results to unit circular arcs and proved hardness results for intervals of arbitrary length and higher-dimensional objects.
We retain the LP formulations in \Cref{sec:lp} as a common approach to these three classes and $k$-connectivity, and include other cases not addressed in these works.
\subsection{Our Contribution}
Our results focus mainly on interval graphs.
\Cref{tab:summary} summarises the results for the individual graph classes.
We use the real-RAM model in \Cref{sec:algorithm,sec:algorithm_kclique}, and we assume rational input and use the bit model in \Cref{sec:lp}.
In the preliminary version \citep*{HonoratoDroguett2024} we gave linear-time algorithms for \gged[\complete] on weighted intervals and on weighted unit squares under the $L_1$ distance, and an $O(n\log n)$-time algorithm for \gged[\kclique] on unweighted unit intervals.
We also gave $O(n)$-size linear programs for the remaining classes in \Cref{tab:summary} on unweighted unit intervals.
\begin{table}[bt]
  \centering
  \caption{Summary of our results for {\gged[]}. `\#Weights' is the number of distinct weights (`1' denotes unit weights). For rectangles, we use the $L_1$ distance. For edgeless, acyclic and $k$-clique-free graphs, we assume that intervals are open.}\label{tab:summary}
  \begin{tabular}{@{}l@{\quad}l@{~~}c@{\quad}c@{\quad}c@{}}
    \toprule
    \bf Graph Class & \bf Object Type & \bf \#Weights & \bf Complexity & \bf Reference                   \\\midrule
    \complete       & Interval        & $n$           & $O(n)$         & \Cref{thm:complete_linear}      \\
                    & Rectangle       & $n$           & $O(n)$         & \Cref{thm:rectangles_linear}    \\
    \kclique        & Unit Interval   & 1             & $O(n\log n)$   & \Cref{thm:kclique_uig_uniquemd} \\
    \midrule
    \nokclique      & Unit Interval   & 1             & \poly~(LP)     & \Cref{thm:nokclique:lp}         \\
    \edgeless       & Unit Interval   & 1             & \poly~(LP)     & \Cref{cor:edgeless_acyclic:lp}  \\
    \acyclic        & Unit Interval   & 1             & \poly~(LP)     & \Cref{cor:edgeless_acyclic:lp}  \\
    \kconnected     & Unit Interval   & 1             & \poly~(LP)     & \Cref{thm:kconnected:lp}        \\
    \bottomrule
  \end{tabular}
\end{table}
\paragraph{Extensions of the Conference Version}
We extend the unit-square result to weighted axis-parallel rectangles of arbitrary size (\Cref{ssec:results_usg}).
We simplify the $k$-clique algorithm by sweeping the interval endpoints in linear time after sorting, retaining the $O(n\log n)$ overall bound (\Cref{alg:kclique-2}).
After presenting our prune-and-search algorithm for \gged[\complete] on weighted intervals in the preliminary version of this paper, we noted that this problem reduces to a weighted median of the interval endpoints, which also gives a linear-time algorithm (\Cref{ssec:geometric_median}).
We keep our algorithm as an alternative approach for our model, with the pruning steps and weighted cost updates made explicit (\Cref{alg:complete,lem:alg:opt}), because we can extend it to maximising the common intersection length of weighted intervals within a movement threshold in $O(n)$ time (\Cref{thm:max_overlap_linear}).
For the LP formulations, we add explicit centre-order constraints and correct the indexing in the $k$-connectivity case.
We further combine these formulations to obtain an $r$-connected, $k$-clique-free graph in polynomial time (\Cref{thm:connected_nokclique:lp}).
When $k=r+2$, the resulting graph is the $r$th power of a path (\Cref{cor:path_power:lp}).
We also combine the connectivity formulation with dynamic programming to obtain a graph with at most $p$ connected components in polynomial time (\Cref{thm:components:dp}).

\subsubsection{Paper Organisation}

This paper is organised as follows.
\Cref{sec:preliminaries} formally describes the definitions needed to address the above ideas.
In \Cref{sec:algorithm},
we show that \gged[\complete] is solvable in (i) $O(n)$ time on weighted interval graphs and (ii) $O(n)$ time on weighted rectangle graphs over the $L_{1}$ distance.
We also maximise the common intersection length of weighted intervals within a movement threshold in $O(n)$ time.
Given an $n$-tuple of unit intervals, in \Cref{sec:algorithm_kclique} we give an $O(n \log n)$-time algorithm for \gged[\kclique].
In \Cref{sec:lp}, we give linear programming formulations for obtaining graphs contained in several fundamental graph classes.
As we shall detail, obtaining graphs in $\edgeless$ and $\acyclic$ corresponds to the cases $k=2$ and $k=3$ of $\nokclique$, respectively, on unit interval graphs.
We also combine the formulations for $r$-connectivity and $k$-clique-freeness, and use the connectivity formulation in a dynamic program for obtaining at most $p$ connected components.
Lastly, we present our concluding remarks in \Cref{sec:conclusion}.

\section{Preliminaries}
\label{sec:preliminaries}
In this section, we give general definitions that we use throughout the paper.
We refer to the basic terminology of graphs, linear programming, and convex functions described in the textbooks \citep*{boyd2004,cormen2009,Diestel:BOOK:2010}.

In subsequent sections, we use the $L_2$ distance (also known as the Euclidean distance) and the $L_1$ distance (also known as the Manhattan distance).
Each object has a positive \emph{distance weight}.
Unless stated otherwise, all distance weights are one, and the intersection graph is \emph{unweighted}.

\paragraph{Objects}
We distinguish objects by their indices, even when they coincide geometrically.
Set notation on tuples refers to their indexed entries.
Given two points $a,b\in\mathbb{R}$ with $a\le b$, the \emph{interval} defined by $a$ and $b$ is $I=[a,b]$.
We write $\ell(I)=a$ and $r(I)=b$ and call them the \emph{left endpoint} and the \emph{right endpoint} of $I$, respectively.
Its length is $b-a$, and its \emph{centre} is $c(I)=(a+b)/2$.
An interval of length $1$ is called a \emph{unit interval}.
A \emph{rectangle} $S$ is a closed axis-parallel rectangle with horizontal and vertical side lengths $\mathrm{len}_x(S),\mathrm{len}_y(S)>0$, respectively, and its centre is positioned at $(s_x,s_y) \in \mathbb{R}^2$.
The \emph{left interval set} of a tuple of intervals $\I$ is the subcollection of intervals to the `left' of a given point $x$, denoted as $L(\I,x)$.
That is, $L(\I, x) = \{I \in \I:\: r(I) < x\}$.
Similarly, the \emph{right interval set} is defined as $R(\I,x) = \{I\in \I:\:\ell(I) > x\}$.
Moreover, the \emph{leftmost endpoint} $\ell(\I)$ of a tuple of intervals $\I$ is defined as $\min_{I \in \I} \ell(I)$.
Similarly, the \emph{rightmost endpoint} $r(\I)$ of $\I$ is defined as $\max_{I \in \I} r(I)$.
The \emph{endpoint multiset} $\E(\I)$ of $\I$ contains the endpoints $\ell(I)$ and $r(I)$ for every $I \in \I$.
Whenever we write $\E(\I)=(e_1,\ldots,e_{2|\I|})$, the endpoints are ordered by the natural order, so $e_i \le e_{i+1}$ for every $1\le i<2|\I|$.

\paragraph{Graphs}
We consider only simple, finite, and undirected graphs.
Let $\S=(S_1,\ldots,S_n)$ be an $n$-tuple of geometric objects in $\mathbb{R}^d$ (non-empty subsets of $\mathbb{R}^d$).
For any $S,S' \in \S$, we say that $S$ and $S'$ \emph{intersect} when $S \cap S' \neq \emptyset$.
The \emph{intersection graph} of $\S$, denoted by $G_\S$, has one vertex for each entry of $\S$, and an edge between two distinct vertices if and only if their corresponding objects intersect.
When $\S$ is a tuple of (unit) intervals, $G_\S$ is called a \emph{(unit) interval graph}.
Similarly, $G_\S$ is called a \emph{rectangle graph} when $\S$ is a tuple of rectangles.

For an integer $k\ge 1$, a graph $G = (V,E)$ is \emph{$k$-connected} if $|V|>k$ and $G$ remains connected after removing any set of fewer than $k$ vertices.
In this paper, we mainly deal with the following classes:
(i)   $\complete$: complete graphs,
(ii)  $\edgeless$: edgeless graphs,
(iii) $\acyclic$: acyclic graphs,
(iv)  $\kclique$: graphs containing a $k$-clique,
(v)   $\nokclique$: graphs containing no $k$-clique, and
(vi)  $\kconnected$: $k$-connected graphs.
We also consider the class $\kconnected[r]\cap\nokclique$ and graphs with at most $p$ connected components.

\section{Linear-time Algorithms for Obtaining Complete Graphs}
\label{sec:algorithm}
This section presents two linear-time algorithms for \gged[\complete], one on weighted interval graphs and one on weighted rectangle graphs over the $L_1$ distance.
We then adapt the first algorithm to maximise the common intersection length of weighted intervals within a movement threshold in linear time.
We assume that the input is not given sorted.
\subsection{The Algorithm for Weighted Intervals}\label{ssec:winterval_linear}
We use a prune-and-search approach.
For an arbitrary point $x$,
the minimum movement that ensures that all intervals contain $x$ is uniquely determined by moving the intervals in $L(\I,x)$ to the right and moving the intervals in $R(\I,x)$ to the left.
We define this movement by a function that gives the total moving distance of the intervals, and call it the \emph{moving distance function}.
\begin{definition}\label{def:moving_distance}
  The \emph{moving distance} of an interval $I \in \I$ to a point $x$ is the function $d_{I}:\mathbb{R} \rightarrow \mathbb{R}$ defined by
  \begin{gather}
    d_{I}(x) = \begin{cases}
      w_I(x-r(I)),\quad x>r(I)       \\
      w_I(\ell(I)-x),\quad x<\ell(I) \\
      0,\quad \text{otherwise}
    \end{cases},
    \label{equ:moving_distance}
  \end{gather}
  where $w_I \in \mathbb{R}^+$ is the distance weight of $I$. The \emph{total moving distance} of a tuple of intervals is then defined as $D(\I,x) = \sum_{I\in\I} d_I(x)$.
\end{definition}

A point $x \in \mathbb{R}$ where we move all intervals is called a \emph{gathering point}.
When $x$ minimises the function, we call it an \emph{optimal gathering point}.
Since pairwise intersecting intervals have a common point, and every interval must be moved at least its moving distance to this point, finding an optimal gathering point gives an optimal solution for \gged[\complete].
\Cref{fig:md_representation} illustrates the curve of an arbitrary moving distance function. Intuitively, if the interval $I$ is moved from the left to $x$, its moving distance is given by $w_I(x-r(I))$; if it is moved to $x$ from the right, its moving distance is given by $w_I(\ell(I)-x)$.
In all other cases, the interval intersects $x$, and hence its moving distance is $0$.

\begin{figure}[bt]
  \centering
  \includegraphics[scale=1,page=1]{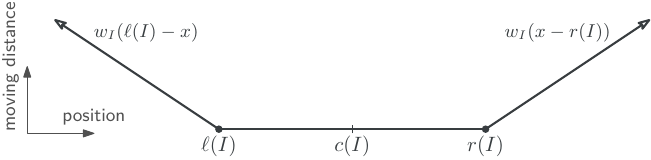}
  \caption{Curve of the moving distance function of an arbitrary interval $I$ with centre $c(I)$.}
  \label{fig:md_representation}
\end{figure}

We state the following properties of these functions.
\begin{lemma}
  Let $\I$ be a nonempty tuple of intervals.
  The moving distance $d_I(x)$ and the total moving distance $D(\I,x)$ are convex and piecewise linear.
  Moreover, there is an optimal gathering point $x$ such that $x$ is an endpoint of an interval in $\I$.
  \label{lem:md_convex}
\end{lemma}
\begin{proof}
  Since $d_I(x)=w_I\max\{\ell(I)-x,0,x-r(I)\}$, it is the maximum of three affine functions and hence convex and piecewise linear.
  Thus $D(\I,x)$ is also convex and piecewise linear, as it is a sum of such functions.
  The function $D(\I,x)$ strictly decreases before $\ell(\I)$ and strictly increases after $r(\I)$.
  Since it is continuous, it has an optimal point in $[\ell(\I),r(\I)]$.
  If an optimal point is not an endpoint value, then it lies in a linear piece of slope zero.
  Therefore an endpoint value of this piece is also optimal.
\end{proof}

In what follows, we use the endpoint values of $\E(\I)$ as the breakpoints of $D(\I,x)$.
Between two consecutive endpoint values, $D(\I,x)$ is linear, and by \Cref{lem:md_convex} the slopes of these linear pieces are non-decreasing.
Since the total moving distance function is convex, the set of optimal gathering points is an interval.
Next, we show how to recognise an optimal gathering point.

\begin{lemma}\label{lem:opt}
  Let $x$ be an endpoint value in $\E(\I)$, and let $x^{-}$ and $x^{+}$ be the previous and next endpoint values of $x$, respectively, when they exist.
  If $D(\I,x) \le D(\I,x^{-})$ whenever $x^{-}$ exists and $D(\I,x) \le D(\I,x^{+})$ whenever $x^{+}$ exists, then $x$ is an optimal gathering point.
\end{lemma}
\begin{proof}
  Since $D(\I,x)$ is convex and linear between consecutive endpoint values, it is enough to compare $x$ with its adjacent endpoint values.
  If $x^{-}$ exists, then $D(\I,x) \le D(\I,x^{-})$ implies that every linear piece to the left of $x$ has non-positive slope.
  Hence $D(\I,y) \ge D(\I,x)$ for any $y < x$.
  If $x^{+}$ exists, then $D(\I,x) \le D(\I,x^{+})$ implies that every linear piece to the right of $x$ has non-negative slope.
  Hence $D(\I,y) \ge D(\I,x)$ for any $y > x$.
  If one adjacent endpoint value does not exist, monotonicity outside $[\ell(\I),r(\I)]$ gives the corresponding inequality.
  Therefore, $x$ is an optimal gathering point.
\end{proof}

\paragraph{Algorithm Overview}
\Cref{alg:complete} finds an optimal gathering point $x$ by narrowing the search to the range $[\ell,r]$.
By \Cref{lem:md_convex}, we initially set $\ell=\ell(\I)$ and $r=r(\I)$.
We keep a current tuple $\I'$, initially equal to $\I$.
The values $W_L$ and $W_R$ store the total weights of intervals removed to the left and right of $[\ell,r]$, respectively.
For each $z\in[\ell,r]$, we add $(z-\ell)W_L+(r-z)W_R$ to $D(\I',z)$ to account for the removed intervals.
The resulting value differs from the original total moving distance by a constant independent of $z$, which does not affect subsequent comparisons.
Before each median selection, we remove from $\I'$ the intervals containing the whole range $[\ell,r]$.
When $\I'$ is nonempty, we set $x$ to a median of the endpoints of $\E(\I')$ in $[\ell,r]$.
Next, we compute $x^{-}$ and $x^{+}$ as the previous and next distinct values of $x$ among these endpoints and the boundaries $\ell$ and $r$, respectively.
If one of these values does not exist, we ignore the corresponding comparison.
In the pseudocode, we then interpret any assignment to the corresponding value $v^{-}$ or $v^{+}$ as $+\infty$.
We then calculate the distance values $v$, $v^{-}$ and $v^{+}$ as in \cref{alg:disc_right_mds} and compare them.
We identify three cases.
If $v \le v^{-}$ and $v \le v^{+}$, then $x$ is an optimal gathering point of $\I$, and we return it as the solution.
If $v^{-} < v$, then an optimal gathering point is between $\ell$ and $x^{-}$ and we replace $r$ by $x^{-}$.
Lastly, if $v^{+} < v$, then an optimal gathering point is between $x^{+}$ and $r$ and we replace $\ell$ by $x^{+}$.
In the last two cases, we repeat the comparison using a median of the remaining endpoints between $\ell$ and $r$.
The algorithm stops when it finds a point satisfying the first case.
If the current tuple is empty, the algorithm returns $\ell$ when $W_L\ge W_R$, and $r$ otherwise.

\setlength{\intextsep}{1\baselineskip}
\DontPrintSemicolon
\begin{algorithm}[tb]
  \caption{A linear-time algorithm for \gged[\complete].}
  \label{alg:complete}
  \Procedure{\rm{FindingOptimalGatheringPoint}($\I$)}{
    $\I' \gets \I$;\
    $\ell \gets \ell(\I)$ and $r \gets r(\I)$;\
    $W_L \gets 0$, $W_R \gets 0$\;
    \While{true\label{alg:while}}{
      $\I' \gets \I' \setminus \set{I\in \I'\mid \ell(I) \le \ell \land r\le r(I)}$\;\label{alg:remove_itvs_too_big}
      \lIf{$|\I'|=0$}{\Return $\ell$ if $W_L\ge W_R$, and $r$ otherwise.}
      $x$ is a median of $\set{e\in \E(\I'):\: \ell\le e \le r}$\;\label{alg:disc_right_median}
      $x^{-}$ and $x^{+}$ are the previous and next distinct values of $x$ in $\set{e \in \E(\I') :\: \ell \le e \le r}\cup\{\ell,r\}$\;\label{alg:disc_right_median_l_r}
      $v \gets D(\I',x)+(x-\ell)W_L+(r-x)W_R$, $v^{-} \gets D(\I',x^{-})+(x^{-}-\ell)W_L+(r-x^{-})W_R$, $v^{+} \gets D(\I',x^{+})+(x^{+}-\ell)W_L+(r-x^{+})W_R$\;\label{alg:disc_right_mds}
      \lIf{$v\le v^{-} \text{ and } v\le v^{+}$}{\Return $x$}
      \If{$v^{-} < v$}{
        $r \gets x^{-}$, $W_R \gets W_R+\sum_{I\in R(\I',r)}w_I$, $\I' \gets \I' \setminus R(\I',r)$\;\label{alg:disc_right}
      }
      \Else{
        $\ell \gets x^{+}$, $W_L \gets W_L+\sum_{I\in L(\I',\ell)}w_I$, $\I' \gets \I' \setminus L(\I',\ell)$\;\label{alg:left}
      }
    }
  }
\end{algorithm}

We now prove the correctness of \Cref{alg:complete}.
By \Cref{lem:md_convex,lem:opt}, it is enough to find an endpoint value whose total moving distance is no larger than the values at its adjacent endpoint values.
We show that \Cref{alg:complete} finds such a point.
To do so, we use \Cref{obs:shift_further_left_right} in subsequent lemmas (see also \Cref{fig:shift_further_left_right}):

\begin{observation}\label{obs:shift_further_left_right}
  Let $\I$ be a nonempty tuple of intervals, and let $x \in \mathbb{R}$.
  If $x\le \ell(\I)$, then $D(\I,y) = D(\I,x) + (x-y)\sum_{I\in \I} w_I$ for any $y \le x$.
  If instead $r(\I) \le x$, then $D(\I,y) = D(\I,x) + (y-x)\sum_{I\in \I} w_I$ for any $x\le y$.
\end{observation}
\begin{proof}
  Assume that $x\le \ell(\I)$.
  For any $y\le x$ and $I\in\I$, we have $d_I(y)=w_I(\ell(I)-y)=d_I(x)+w_I(x-y)$.
  Summing over $\I$ gives the first equality.
  The second follows analogously.
\end{proof}

\begin{figure}[bt]
  \centering
  \includegraphics[scale=1,page=1]{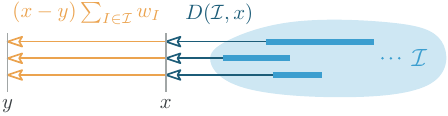}
  \caption{Illustration of \Cref{obs:shift_further_left_right}: The total moving distance of $\I$ to $y$ can be described as the total moving distance of $\I$ to $x$ plus the moving distance of the shifting from $x$ to $y$.}
  \label{fig:shift_further_left_right}
\end{figure}

\begin{lemma}\label{lem:alg:opt}
  Let $x$ be a point returned by \Cref{alg:complete} for a tuple $\I$.
  Then $x$ is an optimal gathering point of $\I$.
\end{lemma}
\begin{proof}
  Consider an iteration of \Cref{alg:complete}, and let $\I'$, $\ell$, $r$, $W_L$ and $W_R$ be the current values of the variables.
  We show that the following two properties hold throughout the algorithm.
  First, there is a constant $C$, independent of $z$, such that $D(\I,z)=D(\I',z)+(z-\ell)W_L+(r-z)W_R+C$ for every $z$ with $\ell\le z\le r$.
  Second, the range $[\ell,r]$ contains an optimal gathering point of $\I$.
  Both hold initially with $\I'=\I$, $W_L=W_R=0$ and $C=0$.

  The deletion in \cref{alg:remove_itvs_too_big} preserves the first property, since every deleted interval contains the whole range $[\ell,r]$, and hence contributes zero to the moving distance for every point in this range.
  Suppose that $v^{-}<v$.
  Then the algorithm sets $r=x^{-}$ and removes the subcollection $R(\I',r)$.
  If this subcollection is nonempty, \Cref{obs:shift_further_left_right} gives its contribution to $z$ as $D(R(\I',r),r)+(r-z)\sum_{I\in R(\I',r)}w_I$ for every $z$ with $\ell\le z\le r$.
  Thus adding $\sum_{I\in R(\I',r)}w_I$ to $W_R$ and deleting $R(\I',r)$ from $\I'$ preserves the first property, changing only the constant $C$.
  Since $v^{-}<v$ and there is no endpoint value between $x^{-}$ and $x$, convexity implies that no optimal gathering point in the range before the update lies to the right of $x^{-}$.
  Thus replacing $r$ by $x^{-}$ preserves the second property.
  The case $v^{+}<v$ is analogous.

  By the first property, comparing the stored values $v$, $v^{-}$ and $v^{+}$ is equivalent to comparing the original total moving distance of $\I$ at $x$, $x^{-}$ and $x^{+}$.
  When the algorithm returns $x$, $v\le v^{-}$ and $v\le v^{+}$, so by the same argument as in \Cref{lem:opt}, $x$ minimises $D(\I,\cdot)$ in the current range.
  If the current tuple is empty, then $D(\I',z)+(z-\ell)W_L+(r-z)W_R$ is linear with slope $W_L-W_R$ for $\ell\le z\le r$.
  Hence this function is minimised at $\ell$ when $W_L\ge W_R$, and at $r$ otherwise, so the algorithm returns an optimal point in the current range.
  In both cases, the second property implies that the returned point is an optimal gathering point of $\I$.
\end{proof}
We analyse the time complexity of \Cref{alg:complete}.

\begin{theorem}
  Given an $n$-tuple of intervals, an optimal gathering point and the corresponding minimum total moving distance can be found in $O(n)$ time.
  \label{thm:complete_linear}
\end{theorem}
\begin{proof}
  By \Cref{lem:alg:opt}, \Cref{alg:complete} returns an optimal gathering point, so it remains to bound its runtime.
  We assume that $\I$ is given as an unsorted doubly linked list.
  Each iteration scans the current list a constant number of times to delete the intervals containing $[\ell,r]$, find a median and its adjacent candidate values by linear-time selection \citep*{Blum1973}, evaluate the three distances, and delete the intervals in \cref{alg:disc_right} or \cref{alg:left}.
  Let $m$ be the number of endpoints of $\E(\I')$ in $[\ell,r]$ after \cref{alg:remove_itvs_too_big}.
  Every remaining interval intersects $[\ell,r]$ without containing it, since each update removes the intervals outside the new range, so it has an endpoint in $[\ell,r]$ and $|\I'|\le m$.
  In particular, the median exists whenever $\I'$ is nonempty.
  The next iteration starts with a subcollection of the current $\I'$, so it takes $O(m+1)$ time.
  Since $x$ is a median and the next range excludes $x$, the next value of $m$ is at most half of the current one.
  Therefore the total time is $O(2n(1+1/2+1/4+\cdots))=O(n)$.
  Finally, one scan of $\I$ gives $D(\I,x)$.
\end{proof}

\subsubsection{All distance weights are the same}
In the definition of the moving distance, we allow each interval $I$ to have a distinct distance weight $w_I$.
We now consider the case where $w_I = w_J$ for any $I,J\in\I$.
In this case, we say that the \emph{moving distance function is uniform}.

\begin{theorem}
  If all intervals in $\I$ have the same distance weight, then every point in $[e_n,e_{n+1}]$ is an optimal gathering point, where $\E(\I)=(e_1,\ldots,e_{2n})$.
  \label{thm:ogp_median}
\end{theorem}
\begin{proof}
  Let $w$ be the value such that $w_I=w$ for every $I\in \I$.
  To the left of all endpoint values, the slope of $D(\I,x)$ is $-nw$.
  When passing an endpoint of $\E(\I)$ from left to right, the slope increases by $w$.
  Hence if $e_n<e_{n+1}$, the slope is $0$ between $e_n$ and $e_{n+1}$.
  Therefore, $D(\I,x)$ is non-increasing up to $e_n$, constant on $[e_n,e_{n+1}]$ if $e_n<e_{n+1}$, and non-decreasing after $e_{n+1}$.
  If $e_n=e_{n+1}$, this interval is a single point.
  Thus every point in $[e_n,e_{n+1}]$ is an optimal gathering point.
\end{proof}

\subsection{The Algorithm for Weighted Rectangles}\label{ssec:results_usg}
Given an $n$-tuple of weighted rectangles, we show that \gged[\complete] over the $L_1$ distance can be solved in $O(n)$ time.
For each rectangle $S$, we define the \emph{moving distance function $d_{S}$} as follows.
\begin{definition}
  The \emph{moving distance} of a rectangle $S \in \S$ with centre $(s_x, s_y)$ to a point $p = (p_x, p_y)$ is the function $d_{S}:\mathbb{R}^2 \rightarrow \mathbb{R}$ defined by
  \[
    d_{S}(p) = w_S(\max(|s_x - p_x| - \mathrm{len}_x(S)/2, 0) + \max(|s_y - p_y| - \mathrm{len}_y(S)/2, 0)),
  \]
  where $w_S \in \mathbb{R}^+$ is the distance weight of $S$.
  The \emph{total moving distance} of a tuple of rectangles is then defined as $D(\S, p) = \sum_{S\in\S} d_S(p)$.
\end{definition}

Similarly as for intervals, the point $p$ is an \emph{optimal gathering point} if $p$ minimises the total moving distance function.

\begin{theorem}\label{thm:rectangles_linear}
  Given an $n$-tuple of weighted rectangles $\S$,
  an optimal gathering point and the corresponding minimum total moving distance can be found in $O(n)$ time.
  This gives an optimal solution for \gged[\complete] over the $L_{1}$ distance.
\end{theorem}
\begin{proof}
  Let $\S = (S_1, \ldots, S_n)$ be an $n$-tuple of weighted rectangles.
  We define the $n$-tuples of intervals $\I_x$ and $\I_y$ as the projections of the rectangles in $\S$ onto the $x$- and $y$-axes, respectively, with each interval having the distance weight of its corresponding rectangle.
  Let $p_x$ and $p_y$ be optimal gathering points of $\I_x$ and $\I_y$, respectively.
  We show that the point $p = (p_x, p_y)$ minimises the total moving distance for $\S$.

  By the definition of $d_S$, the contribution of each rectangle is the sum of the contribution of its $x$-interval and its $y$-interval.
  Thus, for any point $q=(q_x,q_y)\in\mathbb{R}^2$, $D(\S,p)=D(\I_x,p_x)+D(\I_y,p_y) $ is at most $D(\I_x,q_x)+D(\I_y,q_y)=D(\S,q)$, where the inequality follows from the optimality of $p_x$ and $p_y$.
  If every two rectangles intersect, then every two projected intervals intersect on each axis, so all projected intervals on that axis have a common point.
  Since the rectangles then have a common point, and every rectangle must be moved at least its moving distance to this point, $p$ gives an optimal solution for \gged[\complete].
  By \Cref{thm:complete_linear}, we can find $p_x$ and $p_y$ in $O(n)$ time.
  We can then calculate the value $D(\S,p)$ by one scan of $\S$.
  Therefore, an optimal gathering point and the corresponding minimum total moving distance can be found in $O(n)$ time.
\end{proof}

\subsection{Maximum Overlap within a Movement Threshold}\label{ssec:max_overlap_linear}
Instead of merely obtaining a complete intersection graph of weighted intervals, one might be interested in maximising its {\em tolerance}, that is, the {\em length} of the common intersection of intervals \citep*{Farrahi1999}.
This is desirable when perturbations affect movements, since intervals that intersect only at an endpoint may become disjoint after arbitrarily small shifts of their centres.
In this section we show how to adapt our algorithm to maximise the length of the common intersection of the intervals, within a given {total moving distance}.

Let $T\ge 0$ be the \emph{movement threshold} and let $\omega_{\max}=\min_{I\in\I}(r(I)-\ell(I))$ be the length of the smallest interval.
For $0\le b-a\le \omega_{\max}$, an interval $I$ contains $[a,b]$ after moving by any distance in $[b-r(I),a-\ell(I)]$.
We choose the distance closest to zero, so $I$ moves left by $\max\{\ell(I)-a,0\}$ or right by $\max\{b-r(I),0\}$, and never both since $b-a\le r(I)-\ell(I)$.
We thus formulate the total moving distance as $D(a,b)=D_{\mathrm{left}}(a)+D_{\mathrm{right}}(b)$, where $D_{\mathrm{left}}(a)=\sum_{I\in\I}w_I\max\{\ell(I)-a,0\}$ and $D_{\mathrm{right}}(b)=\sum_{I\in\I}w_I\max\{b-r(I),0\}$ are the total weighted movements to the left and right, respectively.
As the length of the common intersection is bounded by $\omega_{\max}$, we maximise $b-a$ subject to $D(a,b)\le T$ and $0\le b-a\le \omega_{\max}$.

\begin{theorem}\label{thm:max_overlap_linear}
  Given $n\ge 1$ weighted intervals and a movement threshold $T\ge 0$, we can find a distance vector maximising their common intersection length, or report infeasibility, in $O(n)$ time.
\end{theorem}
\begin{proof}
  For a fixed length $\omega\in[0,\omega_{\max}]$, we apply \Cref{alg:complete} to the intervals $[\ell(I),r(I)-\omega]$ with their original weights (for $\omega=\omega_{\max}$, the intervals of length $\omega_{\max}$ become single points, which the algorithm also handles).
  Gathering these intervals at $a$ is equivalent to making the original intervals contain $[a,a+\omega]$, with total moving distance $D(a,a+\omega)$.
  We use $\omega=0$ to check feasibility, reporting infeasibility if the minimum total moving distance exceeds $T$, and $\omega=\omega_{\max}$ to return its solution if its minimum total moving distance is at most $T$.
  For other cases, we assume that the optimal length is less than $\omega_{\max}$, and find it by pruning the endpoints that determine the two ends of the segment at the same time, instead of trying many values of $\omega$.

  Let $a_{\min}=\min_{I\in\I}\ell(I)$, $a_{\max}=\max_{I\in\I}\ell(I)$, $b_{\min}=\min_{I\in\I}r(I)$ and $b_{\max}=\max_{I\in\I}r(I)$.
  We search for $a$ in $[a_{\min},a_{\max}]$ and for $b$ in $[b_{\min},b_{\max}]$.
  An optimal segment shorter than $\omega_{\max}$ must have $a$ and $b$ in these ranges for the following reason.
  First, $a\le a_{\max}$ and $b\ge b_{\min}$ must hold since otherwise we can extend it slightly without increasing its total moving distance.
  Every interval has length at least $\omega_{\max}$, so $b_{\min}\ge a_{\min}+\omega_{\max}$ and $b_{\max}\ge a_{\max}+\omega_{\max}$.
  Thus $a\ge b-\omega_{\max}\ge a_{\min}$ and $b\le a+\omega_{\max}\le b_{\max}$, so $a\in[a_{\min},a_{\max}]$ and $b\in[b_{\min},b_{\max}]$.

  We keep two lists in the pruning, of the left endpoints strictly inside the range of $a$ and of the right endpoints strictly inside the range of $b$, counting endpoints with the same value separately.
  We store, for the original tuple, $\sum_{\ell(I)>t}w_I$ and $\sum_{\ell(I)\ge t}w_I$ for $t\in\{a_{\min},a_{\max}\}$, and $\sum_{r(I)<t}w_I$ and $\sum_{r(I)\le t}w_I$ for $t\in\{b_{\min},b_{\max}\}$.
  We will use these values in the comparisons below.
  To evaluate $D(a,b)$, we also store $\sum_{\ell(I)\ge a_{\max}}w_I\ell(I)$ and $\sum_{r(I)\le b_{\min}}w_Ir(I)$.
  With the stored weights, the first gives the contribution $\sum_{\ell(I)\ge a_{\max}}w_I\ell(I) -a\sum_{\ell(I)\ge a_{\max}}w_I$ of the removed left endpoints to $D_{\mathrm{left}}(a)$, and the second gives the analogous contribution to $D_{\mathrm{right}}(b)$.
  Other removed endpoints contribute zero.
  We update all stored values whenever we replace a range end.

  While a current list is nonempty, we choose a median from the list with more endpoints (if there are two medians, we choose one arbitrarily).
  Let $p$ be a median of the left endpoints, and let $\lambda=\sum_{\ell(I)>p}w_I$ be the weight of the intervals strictly to the right of $p$.
  We find $q\in[b_{\min},b_{\max}]$, choosing it to balance movement in the two directions, so that $\sum_{\ell(I)>p}w_I\le\lambda\le\sum_{\ell(I)\ge p}w_I$ and $\sum_{r(I)<q}w_I\le\lambda\le\sum_{r(I)\le q}w_I$.
  These inequalities ensure that the segment $[p,q]$ cannot be moved to either side to obtain a better solution.
  If the median is a right endpoint $q$, we proceed analogously with $\lambda=\sum_{r(I)<q}w_I$, the weight of the intervals strictly to the left of $q$, and find $p$ instead.

  The following claim shows that the choice above is well defined.
  \begin{claim}\label{claim:max_overlap_q}
    For every median $p$ of the left endpoints, there is a $q\in[b_{\min},b_{\max}]$ with $\sum_{r(I)<q}w_I\le\lambda\le\sum_{r(I)\le q}w_I$.
  \end{claim}
  \begin{proof}
    \claimqed
    The pairs $(a_{\max},b_{\min})$ and $(a_{\min},b_{\max})$ satisfy the two weight inequalities, with $\lambda=0$ and $\lambda=\sum_Iw_I$ respectively at the start.
    When we update the ranges, we replace one of these pairs by $(p,q)$, which satisfies the inequalities by construction, so the two pairs keep satisfying them.
    Since $a_{\min}<p<a_{\max}$, we obtain $\sum_{r(I)<b_{\min}}w_I \le\sum_{\ell(I)\ge a_{\max}}w_I \le\lambda \le\sum_{\ell(I)>a_{\min}}w_I \le\sum_{r(I)\le b_{\max}}w_I$.
    As the sums over right endpoints increase with $q$, there is a $q\in[b_{\min},b_{\max}]$ with $\sum_{r(I)<q}w_I\le\lambda\le\sum_{r(I)\le q}w_I$.
  \end{proof}

  To find a $q$ that satisfies $\sum_{r(I)<q}w_I\le\lambda\le\sum_{r(I)\le q}w_I$, we use the prune idea of \Cref{alg:complete}.
  At each step, we take a median of the candidates, which are initially the endpoints in the current right list and the two ends of its range, using the linear-time procedure of \citet*{Blum1973}.
  We then prune depending on how $\lambda$ compares with the weight of the right endpoints strictly below the median and with the weight of those up to it, and the median is $q$ if $\lambda$ lies between them.
  These weights follow from the stored sums and the weights of the endpoints excluded from the search, which we do not remove from the list.

  \begin{claim}\label{claim:max_overlap_shift}
    Among the segments of a given length, there is one of minimum total moving distance that contains $[p,q]$ if the length is at least $q-p$, and lies inside $[p,q]$ if it is at most $q-p$.
  \end{claim}
  \begin{proof}
    \claimqed
    The weight inequalities let us move a segment towards $p$ and $q$ without increasing its total moving distance.
    If $a<p$ and $b<q$, we shift both endpoints right by $s=\min\{p-a,q-b\}$.
    The movement to the left decreases by at least $s\sum_{\ell(I)\ge p}w_I\ge s\lambda$, while the movement to the right increases by at most $s\sum_{r(I)<q}w_I\le s\lambda$.
    If $a>p$ and $b>q$, we shift both endpoints left analogously.
    These shifts stay within the current ranges.
    After such a shift, $a=p$ or $b=q$, so in all remaining cases the segment contains $[p,q]$ or lies inside it depending on its length.
  \end{proof}

  We now compare $q-p$ with the possible lengths to decide from which side we narrow the ranges.
  If $q-p<0$, every feasible length is larger, and if $q-p>\omega_{\max}$, every feasible length is smaller.
  Otherwise we calculate $D(p,q)$.
  If $D(p,q)\le T$, then $[p,q]$ is feasible and the optimal length is at least $q-p$.
  If $D(p,q)>T$, no feasible segment has length at least $q-p$, since by \Cref{claim:max_overlap_shift} it could be changed into a segment containing $[p,q]$ with no larger total moving distance, which is at least $D(p,q)$.
  So if $q-p<0$ or $D(p,q)\le T$, we replace $a_{\max}$ by $p$ and $b_{\min}$ by $q$.
  We replace $a_{\min}$ by $p$ and $b_{\max}$ by $q$ if instead $q-p>\omega_{\max}$ or $D(p,q)>T$.

  After each replacement, we remove endpoints outside or at the new range ends and update the stored sums.
  Since the median becomes a range end, we remove at least half the larger list, leaving at most three quarters of all current endpoints.
  Each iteration takes time linear in the number of current endpoints, so the total time is $2n(1+3/4+(3/4)^2+\cdots)=O(n)$.
  When both lists are empty, the stored contributions give $D(a,b)$ throughout the remaining ranges, so all constraints are linear.
  We maximise $b-a$ subject to these ranges, $D(a,b)\le T$, and $0\le b-a\le \omega_{\max}$ by checking the feasible intersections of the constant number of constraint lines.
  Finally, one scan of the original intervals recovers the distance vector by choosing, for each interval, the translation closest to zero in $[b-r(I),a-\ell(I)]$.
\end{proof}
Sorting the endpoints gives a simpler $O(n\log n)$-time algorithm that sweeps the pairs $(p,q)$ satisfying the weight inequalities in the above proof.
With \Cref{thm:max_overlap_linear}, we avoid this sorting and achieve linear time.

\subsection{Discussion: Geometric Median}\label{ssec:geometric_median}
We can reformulate the total moving distance $D(\I,x)$ as $(1/2)\sum_{I\in\I}w_I(|x-\ell(I)|+|x-r(I)|)
  -(1/2)\sum_{I\in\I}w_I(r(I)-\ell(I))$.
Thus solving \gged[\complete] is the weighted-median selection on $\E(\I)$, assigning weight $w_I$ to each endpoint of $I$.
This also gives an $O(n)$-time algorithm \citep*{cormen2009}, and \Cref{thm:ogp_median} is its special case with uniform distance weights.

Our prune-and-search approach describes the movements to a common point directly and keeps track of the moving distance contributions of intervals removed during the search.
This also allows us to adapt the algorithm to quadratic movement costs, minimising instead $\sum_{I\in\I}w_I(\max\{\ell(I)-x,0,x-r(I)\})^2$, which is convex and piecewise quadratic.
At a median endpoint, the sign of its derivative determines which side of the search range we keep, or gives an optimum when the derivative is zero.
We can accumulate contributions that become quadratic throughout the maintained range in a single quadratic polynomial.
When no interior endpoints remain, we minimise this polynomial over the remaining range.
Each iteration discards at least half of the interior endpoints, so the total runtime remains $O(n)$.

For a prescribed overlap length $\omega$, the median reduction still gives the minimum total moving distance using the endpoints of $[\ell(I),r(I)-\omega]$.
Our extension in \Cref{ssec:max_overlap_linear} also finds the largest $\omega$ within the threshold in $O(n)$ time by pruning both endpoint ranges.
Applying the median reduction to each candidate value of $\omega$ does not give this bound directly.

\section{Obtaining a \texorpdfstring{$k$}{k}-Clique on Unit Interval Graphs in \texorpdfstring{$O(n \log n)$ time}{O(n log n) time}}
\label{sec:algorithm_kclique}
\newcommand{\dcur}[0]{D_{\mathrm{cur}}}
\newcommand{\imin}[0]{\I_{\min}}
\newcommand{\dmin}[0]{D_{\min}}
\newcommand{\xmin}[0]{x_{\min}}

Given an $n$-tuple of unit intervals $\I$ with a uniform moving distance function and an integer $2\le k\le n-1$, we give an $O(n\log n)$-time algorithm to solve \gged[\kclique].
The case $k=n$ corresponds to \gged[\complete], which we solve in \Cref{sec:algorithm}.
Since multiplying all moving distances by the same positive value
does not change the optimal solution, we normalise this common
distance weight to one throughout this section.
We first prove the following lemma:
\begin{lemma}
  Given an $n$-tuple of unit intervals $\I = (I_1,\ldots, I_n)$ ordered by centre, there is an optimal solution for \gged[\kclique] in which the intervals moved to a gathering point $x$ are $I_i,\ldots, I_{i+k-1}$, where $1\leq i \leq n -k + 1$.
  \label{lem:kconsec_kclique}
\end{lemma}
\begin{proof}
  Let $D^{\mathrm{opt}}$ be an optimal solution.
  Since $G_{\I+D^{\mathrm{opt}}}$ contains a $k$-clique and pairwise intersecting intervals have a common point, there is a point $x$ contained in at least $k$ intervals of $\I+D^{\mathrm{opt}}$.
  Moving only $k$ of the corresponding original intervals minimally so that they contain $x$, and leaving all remaining intervals unmoved, gives a feasible solution no more expensive than $D^{\mathrm{opt}}$.
  Hence it is also optimal.
  Among all such optimal choices, choose intervals $J_1=I_{p_1},\ldots,J_k=I_{p_k}$ with $p_1<\cdots<p_k$ that minimise $p_k-p_1$.

  Suppose that $J_1,\ldots,J_k$ are not consecutive in $\I$.
  Then there is an interval $J'=I_q\in \I\setminus\{J_1,\ldots,J_k\}$ such that $p_1<q<p_k$.
  If $x\le c(J')$, then $d_{J'}(x)\le d_{J_k}(x)$ because the intervals are unit intervals and $J'$ is no farther from $x$ than $J_k$.
  Replacing $J_k$ with $J'$ therefore gives an optimal solution whose chosen intervals have smaller index span.
  If $c(J')<x$, then $d_{J'}(x)\le d_{J_1}(x)$ for the same reason.
  Replacing $J_1$ with $J'$ again gives an optimal solution whose chosen intervals have smaller index span.
  Both cases contradict the choice of $J_1,\ldots,J_k$.
  Therefore the chosen intervals must be consecutive.
\end{proof}

The previous lemma implies that examining sequences of $k$ consecutive intervals of $\I$ is sufficient to obtain a $k$-clique using the minimum moving distance.
Furthermore, we deduce from \Cref{thm:ogp_median} that the $k$th and $(k+1)$th endpoints are optimal gathering points of any subcollection of $k$ unit intervals.

Before the details of the algorithm, we prove \Cref{lem:eitoeiplusone} (see also \Cref{fig:eitoeiplusone}):

\begin{lemma}\label{lem:eitoeiplusone}
  For every subcollection $\J\subseteq\I$ and every $1\le i<2n$,
  the total moving distance $D(\J,e_{i+1})$ equals $D(\J,e_i)+(e_{i+1}-e_i) \bigl(|L(\J,e_{i+1})|-|R(\J,e_i)|\bigr)$.
\end{lemma}
\begin{proof}
  The statement holds if $e_{i} = e_{i+1}$.
  Thus we assume that $e_{i} < e_{i+1}$.
  Since $e_i$ and $e_{i+1}$ are consecutive endpoint values of $\E(\I)$, no endpoint of $\J$ lies strictly between them.
  Let $\delta$ be equal to $e_{i+1}-e_i$.
  For each $I\in L(\J,e_{i+1})$, the distance of $I$ to $e_{i+1}$ is $d_I(e_{i+1})=d_I(e_i)+\delta$.
  Similarly, for each $I\in R(\J,e_i)$, $d_I(e_{i+1})=d_I(e_i)-\delta$.
  All remaining intervals of $\J$ contain both $e_i$ and $e_{i+1}$, so their moving distance to either endpoint is zero.
  Combining these equalities, we obtain
  $$
    D(\J,e_{i+1})
    =D(\J,e_i)+\delta(|L(\J,e_{i+1})|-|R(\J,e_i)|).
  $$
  This proves the lemma.
\end{proof}

\begin{figure}[t]
  \centering
  \includegraphics[scale=1,page=1]{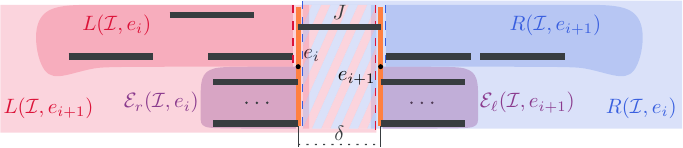}
  \caption{Illustration of \Cref{lem:eitoeiplusone}: Here $\mathcal{E}_r(\mathcal{I},e_i)$ and $\mathcal{E}_\ell(\mathcal{I},e_{i+1})$ are the subcollections of intervals having $e_i$ as the right endpoint and $e_{i+1}$ as the left endpoint, respectively. For $e_i<e_{i+1}$, the subcollection $L(\mathcal{I},e_{i+1})$ consists of the union of $L(\mathcal{I},e_{i})$ and $\mathcal{E}_r(\mathcal{I},e_{i})$. Similarly, $R(\mathcal{I},e_{i})$ consists of the union of $R(\mathcal{I},e_{i+1})$ and $\mathcal{E}_\ell(\mathcal{I},e_{i+1})$. Intervals containing both $e_i$ and $e_{i+1}$ are not shown in these subcollections, since their moving distance to both endpoints is zero.}
  \label{fig:eitoeiplusone}
\end{figure}

\Cref{lem:eitoeiplusone} allows us to calculate the total moving distance to an endpoint in constant time when we have the total moving distance to the previous endpoint, once we know the two counts of intervals.

\paragraph{Algorithm Overview} Given an $n$-tuple of unit intervals $\I$ and a value $k$, the algorithm sweeps $\I$ looking for the subcollection of $k$ consecutive intervals $\imin \subseteq \I$ with optimal gathering point $\xmin$.
The algorithm starts by sorting $\I$ and determining the endpoint multiset $\E(\I) = (e_1,\ldots,e_{2n})$.
For $k\le j\le n$, let $\J_j=\set{I_{j-k+1},\ldots,I_j}$.
Once $\I$ is sorted, we initially set $\imin$ to the subcollection $\J_k$ and $\xmin$ to be $e_1$, and calculate the initial total moving distance $\dmin = D(\imin,\xmin)$.
Assume that the current subcollection is $\J_j$ for $k\le j \le n$.
We iterate over $e_1,\ldots,e_{2n}$.
In the $i$th iteration, we advance the current subcollection while its total moving distance to $e_i$ does not increase.
While $j<n$, we compare the distances to the point $e_i$ of the $(j-k+1)$th interval and the $(j+1)$th interval.
If $d_{I_{j-k+1}}(e_i) \ge d_{I_{j+1}}(e_i)$, then we replace $\dcur$ with $\dcur - d_{I_{j-k+1}}(e_{i}) + d_{I_{j+1}}(e_{i})$ and increase $j$ by one.
We repeat this step until $j=n$ or the inequality becomes false.
Once we obtain the subcollection with minimum total moving distance for $e_i$, we compare $\dcur$ with $\dmin$.
If $\dmin > \dcur$, we replace $\imin$, $\dmin$ and $\xmin$ with $\J_j$, $\dcur$ and $e_i$, respectively.
If $i<2n$, we then calculate $D(\J_j,e_{i+1})$ using $D(\J_j,e_{i})$.

\DontPrintSemicolon
\begin{algorithm}[t]
  \caption{An $O(n \log n)$-time algorithm for \gged[\kclique]}
  \label{alg:kclique-2}

  \Procedure{\rm{FindingOptimalClique}($\I$,$k$)}{
    Sort $\I$ such that $c(I_i)~\leq~c(I_{i+1})$ for all $1\leq i \leq n-1$.\;\label{alg2:sort}
    Compute $\E(\I) = (e_1,\ldots,e_{2n})$\;
    Precompute $|L(\I,e_i)|$ and $|R(\I,e_i)|$ for every $1\le i\le 2n$\;
    $\imin \gets \J_k$, $\dmin \gets D(\imin,e_1)$, $\xmin \gets e_1$\;
    $\dcur \gets \dmin$, $j\gets k$\;
    \For{$i\gets 1$ \KwTo $2n$}{\label{alg2:while}
      \While{ $j < n \text{ and } d_{I_{j-k+1}}(e_{i}) \ge d_{I_{j+1}}(e_{i})$}{\label{alg2:while_subcol}
        $\dcur \gets \dcur - d_{I_{j-k+1}}(e_{i}) + d_{I_{j+1}}(e_{i})$\;
        $j\gets j+1$\;
      }
      \lIf{$\dcur < \dmin$}{
        $\imin\gets \J_j,\:\dmin \gets \dcur,\: \xmin \gets e_{i}$
      }
      \lIf{$i<2n$}{
        Calculate $D(\J_j,e_{i+1})$ from $\dcur$ and assign it to $\dcur$\label{alg2:tmd_update}
      }
    }
    \Return $(\imin,\xmin)$\;
  }
\end{algorithm}

In the algorithm, we replace the subcollection $\J_j$ with $\J_{j+1}$ whenever the total moving distance of $\J_{j+1}$ to $e_i$ is no greater than the total moving distance of $\J_j$ to $e_i$.
\Cref{lem:opt_subcol_for_epoint} shows that, after this replacement, we need not consider $\J_j$ at any later endpoint:

\begin{lemma}\label{lem:opt_subcol_for_epoint}
  Let $k\le j<n$.
  If $D(\J_j,x)\ge D(\J_{j+1},x)$ holds, then $D(\J_j,y)\ge D(\J_{j+1},y)$ for every $y\ge x$.
\end{lemma}
\begin{proof}
  Notice that $D(\J_j,x)\ge D(\J_{j+1},x)$ implies that $d_{I_{j-k+1}}(x) \ge d_{I_{j+1}}(x)$.
  For two unit intervals $I_a$ and $I_b$ with $c(I_a)\le c(I_b)$, the function $d_{I_a}(x)-d_{I_b}(x)$ is non-decreasing in $x$.
  Using $d_I(x)=\max\{0,|x-c(I)|-1/2\}$, the slope of $d_{I_a}(x)-d_{I_b}(x)$ is non-negative on each linear piece.
  Since $c(I_{j-k+1})\le c(I_{j+1})$, we have, for every $y\ge x$, $D(\J_j,y)-D(\J_{j+1},y) \ge D(\J_j,x)-D(\J_{j+1},x)\ge 0$,
  since the distances of intervals contained in both $\J_j$ and $\J_{j+1}$ cancel.
\end{proof}

\begin{lemma}\label{lem:kconsec_opt}
  The subcollection $\imin$ is the subcollection of $k$ consecutive intervals with minimum total moving distance. Moreover, $\xmin$ is an optimal gathering point of $\imin$.
\end{lemma}
\begin{proof}
  Fix an endpoint $e_i$ and consider shifting the current subcollection to the right in the sorted order.
  The values $d_{I_1}(e_i),\ldots,d_{I_n}(e_i)$ do not increase while the centres are at most $e_i$ and do not decrease once the centres are at least $e_i$, since the intervals are unit intervals and their centres are sorted.
  When we move from $\J_j$ to $\J_{j+1}$, the change of the total moving distance to $e_i$ is exactly $d_{I_{j+1}}(e_i)-d_{I_{j-k+1}}(e_i)$.
  If this change is positive, the entering interval has its centre to the right of $e_i$, since otherwise both centres are at most $e_i$ and $d_{I_{j-k+1}}(e_i)\ge d_{I_{j+1}}(e_i)$.
  Hence every later entering interval also has its centre to the right of $e_i$, and its distance is at least $d_{I_{j+1}}(e_i)$.
  A later outgoing interval with centre to the right of $e_i$ precedes the entering interval in the sorted order, so its distance is at most the distance of the entering interval.
  A later outgoing interval with centre at most $e_i$ has distance at most $d_{I_{j-k+1}}(e_i)$, which is smaller than $d_{I_{j+1}}(e_i)$.
  Hence shifting the subcollection further to the right cannot decrease the total moving distance.
  Applying \Cref{lem:opt_subcol_for_epoint} to successive replacements shows that every discarded subcollection has no smaller total moving distance to $e_i$ than the current subcollection.
  Thus, when the loop in \cref{alg2:while_subcol} stops, the current subcollection $\J_j$ has minimum total moving distance to $e_i$ among all consecutive $k$-subcollections.
  By \Cref{thm:ogp_median,lem:kconsec_kclique}, some optimal solution gathers $k$ consecutive intervals at an endpoint of $\E(\I)$.
  Thus examining the minimum total moving distance at every endpoint is sufficient.
  If $\dcur$ is less than $\dmin$, we properly replace $\imin$, $\dmin$ and $\xmin$.
  Therefore, when the loop in \cref{alg2:while} is completed, $\imin$ is the subcollection of $k$ consecutive intervals with minimum total moving distance and $\xmin$ its optimal gathering point.
\end{proof}

\begin{theorem}
  Given an $n$-tuple of unit intervals $\I$ with a uniform moving distance function, \gged[\kclique] can be solved in $O(n\log n)$ time.
  \label{thm:kclique_uig_uniquemd}
\end{theorem}
\begin{proof}
  Correctness follows from \Cref{lem:kconsec_opt}.
  It remains to analyse the complexity of \Cref{alg:kclique-2}.
  First, we sort $\I$ in $O(n\log n)$ time.
  We merge the sorted left and right endpoints to determine $\E(\I)$ in $O(n)$ time.
  Similarly, we determine $\imin$ and calculate $\dmin$ in $O(n)$ time.
  We represent the current and best subcollections by their first and last indices.
  Each iteration of \cref{alg2:while_subcol} takes constant time, and there are at most $n-k$ such iterations since $j$ increases by one each time.
  Comparing $\dcur$ with $\dmin$ and updating the best solution also take constant time.
  We only need to check \cref{alg2:tmd_update}, where we apply the formula given in \Cref{lem:eitoeiplusone}.
  To apply the formula, we need to count the intervals of the current subcollection that belong to $L(\I,e_{i+1})$ and $R(\I,e_i)$.
  Sorting unit intervals by their centres also sorts their left and right endpoints.
  Consequently we can precompute $|L(\I,e_i)|$ and $|R(\I,e_i)|$ for every endpoint in $O(n)$ time using two scans of the sorted endpoints.
  For the current subcollection $\J_j$, we then have
  \begin{align*}
    |L(\J_j,e_{i+1})| & = \min\{k,\max\{0,|L(\I,e_{i+1})|-j+k\}\}, \\
    |R(\J_j,e_i)|     & = \min\{k,\max\{0,j+|R(\I,e_i)|-n\}\}.
  \end{align*}
  Thus the two counts, and hence the update in \cref{alg2:tmd_update}, take constant time.
  The loop of \cref{alg2:while} ends after iterating over $\E(\I)$, which has size $2n$.
  Therefore, the algorithm runs in $O(n\log n)$ time.
\end{proof}
If $\I$ is given sorted by centre, \Cref{alg:kclique-2} runs in $O(n)$ time.

The assumptions of unit intervals and a uniform moving distance function are required for \Cref{lem:kconsec_kclique,lem:opt_subcol_for_epoint}.
For arbitrary interval lengths or unequal distance weights, an optimal solution need not consist of $k$ consecutive intervals ordered by their centre.

\section{LP-formulations for Other Fundamental Graph Properties on Unit Interval Graphs}
\label{sec:lp}
This section addresses the remaining classes $\nokclique$, $\edgeless$, $\acyclic$, and $\kconnected$ on unit interval graphs.
We start with $\nokclique$ and derive the cases $\edgeless$ and $\acyclic$ from this formulation.
We show that the problems of obtaining graphs in these classes can be formulated using linear programming.
We assume that the interval centres are rational numbers.
This implies that these problems can be solved in polynomial time in the bit model using polynomial-time algorithms for linear programming \citep*{Karmarkar1984}.
We then combine the formulations for $r$-connectivity and $k$-clique-freeness, and use the connectivity formulation in a dynamic program for obtaining at most $p$ connected components.

In the following formulations, we use absolute values in the objective function.
To obtain linear programs, we replace each $x_i$ by $x_i'-x_i''$ and $|x_i|$ by $x_i'+x_i''$, with $x_i',x_i''\ge 0$.
The reformulation preserves optimal solutions, since in an optimum at least one of $x_i'$ and $x_i''$ is zero.
If both are positive, subtracting their minimum from both preserves $x_i'-x_i''$ and decreases the objective function.

Linear programs only admit weak constraints ($\le$, $\ge$, $=$).
In \Cref{ssec:nokclique}, we treat the intervals as open, so two unit intervals are disjoint if and only if the distance between their centres is at least $1$, and all constraints are weak.\footnote{For instance, this assumption is also made by \citet*{fomin2023}, where unit disks are assumed to have no intersection when \emph{touching}, that is, when the distance between their centres is exactly $2$.}
Elsewhere the intervals are closed, and in \Cref{ssec:connected_nokclique} we compute the infimum of the total moving distance, since the minimum need not be attained.

\subsection{Obtaining \texorpdfstring{$k$}{k}-Clique-Free, Edgeless and Acyclic Graphs}\label{ssec:nokclique}
We show that \gged[\nokclique] can be formulated as a linear program.
Let $\I = (I_1,\ldots,I_n)$ be an $n$-tuple of unit intervals ordered by non-decreasing centre.
In this subsection, the intervals are open, that is, $\ell(I),r(I)\notin I$ for every $I\in\I$.
The order of $\I$ induced by interval centres might change under a distance vector $D$.
However, the following lemma allows us to preserve this order without increasing the total moving distance or changing the resulting multiset of intervals.
\begin{lemma}\label{lem:edgeless:same}
  Let $D=(d_1,\ldots,d_n)$ be a distance vector, and let $I_i,I_j$ be intervals with $c(I_i)\le c(I_j)$ and $c(I_i)+d_i\ge c(I_j)+d_j$.
  Exchanging the final centres of $I_i$ and $I_j$ gives a distance vector with the same multiset of intervals, preserving the order between $I_i$ and $I_j$, and with no larger total moving distance.
\end{lemma}
\begin{proof}
  By replacing $d_i$ and $d_j$ with $c(I_j)-c(I_i)+d_j$ and $c(I_i)-c(I_j)+d_i$, respectively, we exchange the final centres of $I_i$ and $I_j$.
  Since the intervals are unit intervals, the multiset of intervals is unchanged.

  The other components of $D$ are unchanged, so it is enough to show that $|d_i| + |d_j| \ge |c(I_i) - c(I_j) + d_i| + |c(I_j) - c(I_i) + d_j|$.
  The function $f(t)=|t-c(I_i)|-|t-c(I_j)|$ is non-decreasing, since $c(I_i)\le c(I_j)$ and its slopes are $0$, $2$ and $0$ on the regions determined by these two centres.
  Since $c(I_j)+d_j\le c(I_i)+d_i$, we have $f(c(I_j)+d_j)\le f(c(I_i)+d_i)$.
  Substituting the definition of $f$ and rearranging yields the claimed inequality.
\end{proof}

By repeatedly applying \Cref{lem:edgeless:same} to adjacent pairs $I_i,I_{i+1}$ with $c(I_i)+d_i>c(I_{i+1})+d_{i+1}$, we obtain a distance vector whose resulting centres preserve the original order.
The swap process ends because each replacement decreases by one the number of pairs $i<j$ with $c(I_i)+d_i>c(I_j)+d_j$.
The resulting multiset of intervals is unchanged and the total moving distance does not increase, so the argument is independent of the graph property.
To characterise unit interval graphs with no $k$-cliques, we first show \Cref{lem:kcliqueiff}:
\begin{lemma}
  Given an $n$-tuple of closed unit intervals $\I = (I_1,\ldots,I_n)$ ordered by non-decreasing centre, a $k$-clique exists in $G_\I$ if and only if there exists an index $i$ with $1 \le i \le n-k+1$ such that $c(I_{i+k-1}) - c(I_i) \le 1$.
  \label{lem:kcliqueiff}
\end{lemma}
\begin{proof}
  Suppose that there exists a $k$-clique in $G_\I$.
  The clique consists of $k$ vertices, so the corresponding $k$ intervals intersect each other in $\I$.
  Let $I_i$ and $I_j$ be the first and last intervals in $\I$ that correspond to vertices of the clique, respectively.
  Since $I_i$ and $I_j$ intersect, $c(I_j) - c(I_i) \le 1$.
  The interval block $I_i,\ldots,I_j$ contains at least $k$ intervals.
  Since the centres are sorted, $c(I_{i+k-1}) \le c(I_j)$.
  Hence $c(I_{i+k-1}) - c(I_i) \le c(I_j)-c(I_i) \le 1$.

  In the other direction, we assume that $c(I_{i+k-1}) - c(I_i) \le 1$ holds for some $i$ with $1 \le i \le n-k+1$.
  Since the centres are sorted, the distance between the centres of any two intervals in $I_{i},I_{i+1},\ldots,I_{i+k-1}$ is at most $1$.
  It follows that this sequence forms a $k$-clique.
\end{proof}
For open unit intervals, \Cref{lem:kcliqueiff} holds with $c(I_{i+k-1})-c(I_i)<1$, by the same proof.
By this lemma there is a $k$-clique when the distance between $c(I_i)$ and $c(I_{i+k-1})$ is less than $1$, for a $1\leq i \leq n-k+1$.
Conversely, $G_\I$ does not contain a $k$-clique when $I_i$ and $I_{i+k-1}$ do not intersect for all $1\leq i \leq n-k+1$.
Thus we impose the constraint $(c(I_{i+k-1})+x_{i+k-1}) - (c(I_i)+x_i) \ge 1,\: 1\leq i \leq n-k+1$ on every such pair after moving the intervals.
It leads to the following $O(n)$-size linear program:
\begin{equation*}
  \begin{array}{llcl@{}l}
    \text{minimise}   & \displaystyle\sum\limits_{i=1}^{n} |x_i|                          & \\
    \text{subject to} & (c(I_{i+k-1})+x_{i+k-1})-(c(I_i)+x_i) \geq  1,\:1\le i \le n-k+1,   \\
                      & (x_{i+1} + c(I_{i+1})) - (x_i + c(I_i)) \geq 0,\:1\le i \le n-1.
  \end{array}
  \label{equ:nokclique_graph_0}
\end{equation*}

\begin{theorem}\label{thm:nokclique:lp}
  For $k\geq 2$ and open unit intervals, an optimal solution for \gged[\nokclique] can be found in polynomial time by linear programming.
\end{theorem}
\begin{proof}
  We show that the above linear program gives an optimal solution for \gged[\nokclique].
  First, let $D = (d_1,\ldots,d_n)\in\mathbb{R}^n$ be a distance vector such that $G_{\I+D}$ is in $\nokclique$.
  By \Cref{lem:edgeless:same}, we can assume that $D$ preserves the order of $\I$.
  Since $G_{\I+D}$ does not contain a $k$-clique, $(c(I_{i+k-1})+d_{i+k-1}) - (c(I_i)+d_i) \ge 1$ holds for $1 \le i \le n-k+1$ by \Cref{lem:kcliqueiff}.
  Hence $D$ is a feasible solution of the linear program.

  Conversely, let $X = (x_1,\ldots,x_n)$ be an optimal solution of the linear program.
  The order constraints ensure that the centres of $\I+X$ preserve the order of $\I$.
  By \Cref{lem:kcliqueiff}, the separation constraints imply that $G_{\I+X}$ is in $\nokclique$.
  Hence the optimal value of the linear program equals the minimum total moving distance, and an optimal solution of the linear program gives an optimal solution for \gged[\nokclique].
\end{proof}

The cases $\edgeless$ and $\acyclic$ correspond to $k=2$ and $k=3$, respectively.
A graph is \emph{chordal} if it has no induced cycle with four or more vertices.
Since interval graphs are chordal \citep*{Jungnickel2013}, they are acyclic if and only if they contain no $3$-clique.
Thus we conclude the following.
\begin{corollary}\label{cor:edgeless_acyclic:lp}
  For open unit intervals, an optimal solution for \gged[\edgeless] and \gged[\acyclic] can be found in polynomial time by linear programming.
\end{corollary}

\subsection{Obtaining \texorpdfstring{$k$}{k}-Connected Graphs}
In this subsection, we show a linear program for \gged[\kconnected] on unit intervals.
Let $\I = (I_1,\ldots,I_n)$ be an $n$-tuple of unit intervals ordered by non-decreasing centre.
For a $k$ with $1\le k < n$, we first show the following lemma:
\begin{lemma}
  $G_\I$ is $k$-connected if and only if the intervals $I_i$ and $I_{i+k}$ intersect, for all $1\le i \le n-k$.
  \label{lem:kconnectediff}
\end{lemma}
\begin{proof}
  Suppose that a pair of intervals $I_i$ and $I_{i+k}$ are disjoint.
  Then by removing the $k-1$ intervals $I_{i+1}, \ldots, I_{i+k-1}$, the remaining graph becomes disconnected since $I_i$ and $I_{i+k}$ are disjoint and any interval that could bridge them would lie strictly between these intervals.
  Thus $G_\I$ is not $k$-connected.

  Suppose now that $I_i$ and $I_{i+k}$ intersect for all $1\le i\le n-k$.
  Let $S$ be a set of at most $k-1$ intervals removed from $\I$.
  Consider two consecutive intervals $I_a$ and $I_b$ in $\I\setminus S$.
  Then all intervals $I_{a+1},\ldots,I_{b-1}$ belong to $S$, so $b-a-1\le k-1$.
  Hence $b\le a+k$.
  Choose an index $p$ with $1\le p\le n-k$, $p\le a$, and $b\le p+k$.
  This is possible by taking $p=\min\{a,n-k\}$.
  Since $I_p$ and $I_{p+k}$ intersect, we have $c(I_{p+k})-c(I_p)\le 1$.
  The centres are sorted, so $c(I_b)-c(I_a)\le c(I_{p+k})-c(I_p)\le 1$.
  Thus $I_a$ and $I_b$ intersect.
  It follows that $G_{\I\setminus S}$ is connected.
  Since this holds for every set $S$ of at most $k-1$ intervals, $G_\I$ is $k$-connected.
  This completes the proof.
\end{proof}

By \Cref{lem:kconnectediff}, \gged[\kconnected] can be solved in polynomial time by the following linear program:
\begin{equation*}
  \begin{array}{llcl@{}l}
    \text{minimise}   & \displaystyle\sum\limits_{i=1}^{n} |x_i|                         \\
    \text{subject to} & (c(I_{i+k})+x_{i+k})-(c(I_i)+x_i)  \le 1,\:1\le i \le n-k,       \\
                      & (x_{i+1} + c(I_{i+1})) - (x_i + c(I_i)) \geq 0,\:1\le i \le n-1.
  \end{array}
  \label{equ:kconnected_graph_0}
\end{equation*}

\begin{theorem}\label{thm:kconnected:lp}
  For $n>k$, an optimal solution for \gged[\kconnected] can be found in polynomial time by linear programming.
\end{theorem}
\begin{proof}
  We show that the above linear program gives an optimal solution for \gged[\kconnected].
  Let $D = (d_1,\ldots,d_n)$ be a distance vector such that $G_{\I+D}$ is in $\kconnected$.

  By \Cref{lem:edgeless:same}, we can assume that $D$ preserves the order of $\I$.
  Since $G_{\I+D}$ is $k$-connected, $(c(I_{i+k})+d_{i+k})-(c(I_i)+d_i) \le 1$ holds for all $1 \le i \le n-k$ by \Cref{lem:kconnectediff}.
  Hence $D$ satisfies the constraints of the linear program.
  This implies that $D$ is a feasible solution.

  In the other direction, the order constraints ensure that every feasible solution $X=(x_1,\ldots,x_n)$ preserves the order of $\I$.
  By \Cref{lem:kconnectediff}, the remaining constraints imply that $G_{\I+X}$ is in $\kconnected$.
  Hence the optimal value of the linear program equals the minimum total moving distance, and an optimal solution of the linear program gives an optimal solution for \gged[\kconnected].
\end{proof}

\subsection{Obtaining \texorpdfstring{$r$}{r}-Connected, \texorpdfstring{$k$}{k}-Clique-Free Graphs}\label{ssec:connected_nokclique}
We now impose $r$-connectivity and $k$-clique-freeness simultaneously.
Let $\I=(I_1,\ldots,I_n)$ be ordered by non-decreasing centre, with $1\le r<n$ and $k\ge 2$.
By \Cref{lem:kconnectediff}, every $r$-connected unit interval graph contains an $(r+1)$-clique, so there is no feasible solution when $k\le r+1$.
Thus we assume $k\ge r+2$.
Since the intervals are closed, the separation constraints are strict, and the minimum total moving distance need not be attained.
Combining the two previous formulations, we compute the infimum using the following $O(n)$-size linear program, where we replace the separation constraints by the corresponding weak inequalities.
\begin{equation*}
  \begin{array}{llcl@{}l}
    \text{minimise}   & \displaystyle\sum\limits_{i=1}^{n}|x_i|                       \\
    \text{subject to} & (c(I_{i+r})+x_{i+r})-(c(I_i)+x_i)\le 1,\:1\le i\le n-r,       \\
                      & (c(I_{i+k-1})+x_{i+k-1})-(c(I_i)+x_i)\ge 1,\:1\le i\le n-k+1, \\
                      & (x_{i+1}+c(I_{i+1}))-(x_i+c(I_i))\ge 0,\:1\le i\le n-1.
  \end{array}
\end{equation*}

\begin{theorem}\label{thm:connected_nokclique:lp}
  For integers $1\le r<n$ and $k\ge r+2$, the infimum of the total moving distance for \gged[{\kconnected[r]\cap\nokclique}] can be computed in polynomial time by linear programming, together with distance vectors whose total moving distance is arbitrarily close to it.
\end{theorem}
\begin{proof}
  Let $D=(d_1,\ldots,d_n)$ be a distance vector such that $G_{\I+D}$ is $r$-connected and $k$-clique-free.
  By \Cref{lem:edgeless:same}, we can assume that $D$ preserves the order of $\I$ without increasing the total moving distance.
  The resulting multiset of intervals is unchanged, so both graph properties are preserved.
  By \Cref{lem:kconnectediff,lem:kcliqueiff}, $D$ satisfies the connectivity constraints and satisfies the separation constraints strictly.
  Hence $D$ is a feasible solution of the linear program, and the optimal value of the linear program is at most the infimum.

  In the other direction, let $X=(x_1,\ldots,x_n)$ be an optimal solution of the linear program.
  The separation constraints may hold with equality for $X$, in which case $G_{\I+X}$ may contain a $k$-clique.
  For $0<\varepsilon\le 1$, we replace each gap between consecutive centres of $\I+X$ by $(1-\varepsilon)$ times the gap length plus $\varepsilon/r$, keeping the leftmost centre fixed.
  The gaps remain non-negative, so the order is not changed.
  A sum of $r$ consecutive gaps that is at most $1$ remains at most $1$, so the connectivity constraints remain satisfied.
  A sum of $k-1$ consecutive gaps that is at least $1$ becomes greater than $1$ since $(k-1)/r>1$, so the separation constraints hold strictly.
  By \Cref{lem:kconnectediff,lem:kcliqueiff}, the resulting graph is $r$-connected and $k$-clique-free.
  The centres tend to the values of $\I+X$ when $\varepsilon$ tends to $0$, so the total moving distance tends to the optimal value of the linear program.
  The statement follows.
\end{proof}

The \emph{$r$th power of a path} on $n$ vertices has vertex set $\{1,\ldots,n\}$ and an edge between $i$ and $j$ exactly when $1\le |i-j|\le r$.
Setting $k=r+2$ gives the following corollary.
\begin{corollary}\label{cor:path_power:lp}
  For $1\le r<n$, the infimum of the total moving distance so that the resulting intersection graph of an $n$-tuple of unit intervals is isomorphic to the $r$th power of a path can be found in polynomial time by linear programming.
\end{corollary}
\begin{proof}
  If $i<j$ and $j-i\le r$, then $I_i$ and $I_j$ lie in a block of $r+1$ consecutive intervals.
  The connectivity constraints place the extreme final centres of each such block at distance at most $1$, so these pairs intersect.
  When $k=r+2$, the separation constraints (strictly satisfied) make the corresponding vertices of every pair of intervals whose indices differ by at least $r+1$ nonadjacent.
  Thus the resulting graph is the $r$th power of a path.
  On the other hand, the $r$th power of a path on $n>r$ vertices is $r$-connected and has no $(r+2)$-clique, so every distance vector giving this graph is feasible and \Cref{thm:connected_nokclique:lp} applies.
\end{proof}

\subsection{Obtaining Graphs with at Most \texorpdfstring{$p$}{p} Connected Components}
Given unit intervals $\I=(I_1,\ldots,I_n)$ ordered by non-decreasing centre and an integer $p\ge 1$, we show how to obtain a graph with {\em at most} $p$ connected components.
For $1\le i\le j\le n$, let $C(i,j)$ be the minimum total moving distance required to make the intersection graph of $(I_i,\ldots,I_j)$ connected.
For $i<j$, we compute $C(i,j)$ using the connectivity linear program with $k=1$ by \Cref{thm:kconnected:lp}.
A single vertex is connected, so $C(i,i)=0$.
Let $F(t,j)$ be the minimum sum of these costs over all partitions of $(I_1,\ldots,I_j)$ into $t$ non-empty consecutive blocks.
We set $F(0,0)=0$, $F(0,j)=+\infty$ for $j\ge 1$, and $F(t,0)=+\infty$ for $t\ge 1$.
For $1\le t\le\min\{p,n\}$ and $1\le j\le n$, we compute
$$
  F(t,j)=\min_{0\le h<j}\bigl(F(t-1,h)+C(h+1,j)\bigr).
$$
The minimum total moving distance is then $\min_{1\le t\le\min\{p,n\}}F(t,n)$.
\begin{theorem}\label{thm:components:dp}
  Given an $n$-tuple of unit intervals with $n\ge 1$ and an integer $p\ge 1$, the minimum total moving distance so that the resulting intersection graph has at most $p$ connected components can be found in polynomial time.
\end{theorem}
\begin{proof}
  Let $D=(d_1,\ldots,d_n)$ be a distance vector such that $G_{\I+D}$ has at most $p$ connected components.
  By \Cref{lem:edgeless:same}, we can assume that $D$ preserves the order of $\I$ without increasing the total moving distance or changing the number of connected components.
  The connected components then form consecutive blocks: a gap greater than $1$ between consecutive centres separates components, while consecutive intervals with no such gap form a connected graph.
  The total moving distance of each block $(I_i,\ldots,I_j)$ is at least $C(i,j)$.
  Hence the value returned by the dynamic program is at most the cost of $D$.

  On the other hand, consider a partition attaining $\min_{1\le t\le\min\{p,n\}}F(t,n)$ and make each block connected independently using an optimal distance vector.
  Every block induces a connected graph, so the resulting graph has at most $p$ connected components.
  Moreover, intersections between different blocks can only merge components.
  The total moving distance is the value of the dynamic program, so this distance vector is optimal.
  We recover it by storing the minimising indices and the optimal distance vectors for the blocks.

  There are $O(n^2)$ block costs, each computed by a linear program of size $O(n)$.
  The dynamic program has $O(\min\{p,n\}n)$ entries, each requiring $O(n)$ time.
  Together with \Cref{thm:kconnected:lp}, we can conclude the statement.
\end{proof}

\section{Concluding Remarks}
\label{sec:conclusion}
We introduced {\ggedlong} and gave algorithms for obtaining complete graphs and $k$-cliques, together with algorithms for several structural properties of unit interval graphs.
These results show which structural properties of unit interval graphs can be efficiently achieved.
Our pruning approach also maximises the common intersection length of weighted intervals within a movement threshold in linear time (\Cref{thm:max_overlap_linear}).
Several questions remain open:
\begin{itemize}[nosep]
  \item[--] What is the complexity of \gged[\complete] on cyclic domains?
  \item[--] Can we efficiently approximate \gged[\complete] on weighted rectangles or disks under the $L_2$ distance?
  \item[--] Can we maximise the common intersection area of weighted rectangles within a movement threshold in linear time?
  \item[--] Is $\Omega(n\log n)$ time necessary for \gged[\kclique] on unweighted unit intervals, or can it be solved in linear time as \gged[\complete]?
  \item[--] Can we obtain faster combinatorial algorithms for imposing $r$-connectivity and $k$-clique-freeness simultaneously, including the case of obtaining the $r$th power of a path?
  \item[--] Can we compute the connectivity costs of all consecutive blocks more efficiently and improve the dynamic program for obtaining at most $p$ connected components?
  \item[--] Beyond the dispersion results of \citet*{HonoratoDroguett2026}, what is the complexity of obtaining $k$-connected graphs, or graphs with at most $p$ connected components, from unweighted unit circular arcs?
  \item[--] What is the complexity of \gged[\kclique] and \gged[{\kconnected[1]}] on unweighted unit disks?
  \item[--] What is the complexity of obtaining acyclic, $k$-clique-free or $k$-connected graphs from weighted unit intervals?
\end{itemize}

\bibliographystyle{spbasic}
\bibliography{ref}

\ifarxiv\else
  \section*{Statements and Declarations}

  \paragraph{Funding}
  This work was supported by the Japan Society for the Promotion of Science (JSPS) KAKENHI Grant Numbers JP22H00513, JP24K02898, JP25K03077 and JP26K23808, and by the Japan Science and Technology Agency (JST) CRONOS Grant Number JPMJCS24K2.

  \paragraph{Competing Interests}
  The authors have no relevant financial or non-financial interests to disclose.

  \paragraph{Data Availability}
  Data sharing is not applicable to this article as no datasets were generated or analysed during the current study.
\fi

\end{document}